\documentclass{amsart}
\usepackage{amsmath,amssymb,url}
\usepackage{orcidlink}
\usepackage{graphics}
\usepackage{mathpazo}
\usepackage{mathrsfs}
\usepackage{gastex}
\usepackage{calc}
\usepackage{pstricks}
\usepackage{tikz}
\usetikzlibrary{matrix,arrows,decorations.markings}
\usetikzlibrary{calc,positioning,automata,fit,arrows.meta}
\tikzset{arrows={[scale=1]}}
\tikzset{every edge/.style={draw,->,>=Latex,auto}}
\usepackage{url}

\DeclareSymbolFont{rsfscript}{OMS}{rsfs}{m}{n}
\DeclareSymbolFontAlphabet{\mathrsfs}{rsfscript}
\DeclareMathOperator{\dt}{.}

\renewcommand{\thefootnote}{\fnsymbol{footnote}}

\newtheorem{thm}{Theorem}
\newtheorem{prop}[thm]{Proposition}
\newtheorem{lem}[thm]{Lemma}
\newtheorem{rmk}[thm]{Remark}
\newtheorem{cor}[thm]{Corollary}
\newtheorem*{coni}{Conjecture RS}

\newcommand{\sa}{synchronizing DFAs}

\newcommand{\scn}{strongly connected}

\newcommand{\mA}{\mathrsfs{A}}

\newcommand{\ov}{\overline}

\begin{document}

\thispagestyle{empty}

\begin{flushleft}

\Large

\textbf{Synchronization of primitive automata}

\end{flushleft}

\normalsize

\vspace{-2.25mm}

\noindent\rule{\textwidth}{1.5pt}

\vspace{2mm}

\begin{flushright}

\textsc{M.~V.~Volkov\,\orcidlink{0000-0002-9327-243X}}

\vspace{3pt}
\scriptsize
\textit{Institute of Natural Sciences and Mathematics,} \\
\textit{Ural Federal University,} \\
\texttt{m.v.volkov@urfu.ru}

\end{flushright}

\vspace{4mm}

\small

\begin{quote}
\textbf{Abstract.} We exhibit new conditions under which a primitive automaton is synchronizing. In particular, we show that the primitivity of an automaton forces its synchronizability whenever the automaton has either a letter of defect 1 or a word of rank 2.
\end{quote}

\renewcommand{\thefootnote}{\arabic{footnote})}
\setcounter{footnote}{0}

\normalsize

\vspace{2mm}

\section{Background and motivation}
\label{sec:intro}

In this note, a \emph{complete deterministic finite automaton} (DFA) is a pair $\mathrsfs{A}=\langle Q,\Sigma\rangle$ of finite non-empty sets equipped with a map $Q\times\Sigma\to Q$ called the \emph{transition function} of $\mA$. The elements of $Q$ and $\Sigma$ are called \emph{states} and, resp., \emph{letters}. The image of a pair $(q,a)\in Q\times\Sigma$ under the transition function is denoted by $q\dt a$.

\emph{Words over $\Sigma$} are finite sequences of letters (including the empty sequence denoted by $\varepsilon$). The set of all words over $\Sigma$ is denoted by $\Sigma^*$. The transition function of $\mA$ extends to a function $Q\times\Sigma^*\to Q$ (denoted in the same way) by recursion: for every $q\in Q$, we set $q\dt\varepsilon:=q$ and $q\dt wa:=(q\dt w)\dt a$ for all $w\in\Sigma^*$ and $a\in\Sigma$. Every word $w\in\Sigma^*$ induces the transformation $q\mapsto q\dt w$ on the set $Q$. For any non-empty subset $P\subseteq Q$, let $P\dt w:=\{p\dt w\mid p\in P\}$ stand for the image of $P$ under this transformation.

A DFA $\mA=\langle Q,\Sigma\rangle$ is called \emph{synchronizing} if it possesses a \emph{reset} word, that is, a word $w\in\Sigma^*$ such that $q\dt w=q'\dt w$ for all $q,q'\in Q$. The minimum length of reset words of $\mA$ is called the \emph{reset threshold}. Synchronizing automata serve as transparent and productive models of error-tolerant systems in many applications; besides, they appear surprisingly in several branches of pure mathematics. We refer the reader to the chapter~\cite{KV} of the `Handbook of Automata Theory' and the author's recent survey \cite{Vo22} for a quick introduction to the area and an overview of its state-of-the art.

We aim to study the relationship between synchronizability and a property called primitivity. To define the latter, recall that a \emph{congruence} on a DFA $\mA=\langle Q,\Sigma\rangle$ is an equivalence $\rho\subseteq Q\times Q$ such that $(q,q')\in\rho$ implies $(q\dt a,q'\dt a)\in\rho$ for all $q,q'\in Q$ and $a\in\Sigma$. A DFA $\mA$ is said to be \emph{primitive}\footnote{An alternative name for this concept that sometimes appears in the literature is \emph{simple}.} if the equality and the universal relation on its state set are the only congruences on $\mA$. Primitive DFAs naturally came into consideration in algebraic automata theory; see, e.g., \cite{Th70}, and in studying automata via their linear representations; see, e.g., \cite{St10,Ry15,AR16}. They also play a role in applications in which DFAs with some state designated as an initial state and some non-empty set of states designated as final states are used as language recognizers; see, e.g., \cite{RV12}, especially, Proposition~1 therein, which the authors attribute to Zoltan \'Esik.

In general, primitivity and synchronizability are independent properties: there exist primitive DFAs that are not synchronizing and \sa\ that are not primitive. However, the conjunction of these properties defines an affluent class that contains many exciting species of DFAs. Say, several series of \sa\ with reset threshold close to the state number squared are observed to be primitive in \cite{AR16}; these include the famous series of DFAs with $n$ states and reset threshold $(n-1)^2$ discovered by Jan \v{C}ern\'{y} \cite{Ce64}. The literature contains many results that give or can be interpreted as conditions under which primitive DFAs become synchronizing; see, e.g., \cite{AS06,Ne09,ABC13,AC14,ABCea16,RS23}. The present note provides two new such conditions that strengthen some known results and allow us to confirm a conjecture proposed in~\cite{RS23}.

To describe our contribution in more detail, we need some extra notions. Given a DFA $\mA=\langle Q,\Sigma\rangle$, the \emph{rank} of $w$ is the cardinality of $Q\dt w$ and the \emph{defect} of $w$ is the cardinality of the set difference $Q\setminus Q\dt w$. Denote by $\Sigma_0$ the set of all letters of defect 0 in $\Sigma$. Assuming $\Sigma_0\ne\varnothing$, one can consider the DFA $\mA_0:=\langle Q,\Sigma_0\rangle$; observe that all letters in $\Sigma_0$ act on $Q$ as permutations. We call a DFA $\mA$ \emph{permutation-primitive} if the DFA $\mA_0$ is primitive.

So far, most studies on the synchronizability of primitive DFAs have actually dealt with permutation-primitive DFAs; see, e.g., the already cited papers \cite{AS06,Ne09,ABC13,AC14,ABCea16}. Considering permutation-primitive DFAs is very natural from the viewpoint of the theory of permutation groups, and work in this direction has revealed several deep ties with the theory of classical combinatorial configurations (such as Latin squares, Steiner systems, Hadamard matrices, and others); see the survey~\cite{ACS17}. From the present note's perspective, however, the permutation-primitivity may look like an ad hoc condition since there are plenty of primitive DFAs that are not permutation-primitive (for instance, the DFAs with a composite number of states from the \v{C}ern\'{y} series~\cite{Ce64} are such), and, moreover, there exist primitive DFAs without letters of defect 0. Still, it is worthwhile to analyze which synchronizability results established for permutation-primitive DFAs extend to general primitive DFAs.

For instance, Peter Neumann \cite[Lemma 2.4]{Ne09} observed that a permutation-primitive DFA with at least three states is synchronizing whenever it has a letter of rank 2. In Section~\ref{sec:rank2}, we show that so is any primitive DFA with at least three states and a letter of rank 2. It is also known that permutation-primitive DFAs possessing a letter of defect 1 are synchronizing; see, e.g., \cite[Theorem 1]{AC14}\footnote{In the literature, it is common to attribute this result to Igor Rystsov with references to either \cite{Ry95} or \cite{Ry00}. These valuable papers study some kinds of \sa\ but do not seem to deal with any form of primitivity.}. In Section~\ref{sec:main}, we prove that the presence of a letter of defect 1 ensures synchronizability for an arbitrary primitive DFA as was conjectured in~\cite{RS23}. The latter fact is a consequence of a general result (Theorem~\ref{thm:unimodal}) having several other applications.

We have made a fair effort to make this note self-contained to a reasonable extent. In particular, it should be understandable without any familiarity with the theory of permutation groups.

\section{Primitive automata with a unimodal letter}
\label{sec:main}

The \emph{graph} of a DFA $\mathrsfs{A}=\langle Q,\Sigma\rangle$ is the labeled directed graph that has $Q$ as the vertex set and the edge from $q$ to $q'$ labeled $a$ for every $q,q'\in Q$ and $a\in\Sigma$ such that $q\dt a=q'$. Fix a letter $a\in\Sigma$ and remove all edges of the graph of $\mA$ except for those labeled $a$. The remaining graph is denoted $\Gamma_a$, and its weakly connected components are called the $a$-\emph{clusters}.

Observe that $\Gamma_a$ has exactly one outgoing edge for every state in $Q$. Take a state $q\in Q$ and consider the path in $\Gamma_a$ starting at $q$:
\[
q\xrightarrow{a} q\dt a\xrightarrow{a}q\dt a^2\dots\xrightarrow{a} q\dt a^k\cdots.
\]
Since $Q$ is finite, states in this path eventually begin repeating, that is, for some non-negative integer $\ell$ and some integer $m>\ell$, we have $q\dt a^\ell=q\dt a^m$. (Here and below, we adopt the convention that $a^0=\varepsilon$.) In other words, each path in $\Gamma_a$ eventually arrives at a cycle. Hence, each $a$-cluster contains a unique cycle (that can degenerate into a loop) and, perhaps, some trees attached to this cycle at their roots. The least non-negative $\ell$ such that $q\dt a^\ell=q\dt a^m$ for some $m>\ell$ is called the $a$-\emph{height} of $q$. If the $a$-height $\ell$ of $q$ is positive, then the state $q\dt a^{\ell}$ is the root of the tree of $\Gamma_a$ containing $q$ and is called the \emph{root} of $q$. The illustration in Fig.~\ref{fig:cluster} shows an $a$-cluster whose states are marked with their $a$-heights. States connected with dashed lines belong to the same class of the equivalence
\[
\ker a:=\{(q,q')\in Q\times Q\mid q\dt a=q'\dt a\}.
\]

\begin{figure}[h]
\begin{center}
\unitlength=.85mm
\begin{picture}(100,52)(-10,-25)\nullfont
\node(A)(-20,20){2} \node(B)(0,20){1} \node(C)(0,1){2} \node(D)(0,-18){1} \node(E)(40,-18){1}
\node(Y)(100,20){3}
\node(AD)(20,1){0} \node(AB)(20,20){0} \node(BC)(40,20){0}
\node(CD)(40,1){0} \node(AC)(60,20){1} \node(BD)(80,20){2}
\node(AE)(60,1){2} \node(BE)(90,1){3}  \node(X)(80,-18){3}
\drawedge[ELside=r](CD,AD){$a$}
\drawedge(AD,AB){$a$} \drawedge(AB,BC){$a$} \drawedge(BC,CD){$a$}
\drawedge(A,B){$a$}
\drawedge(C,B){$a$}
\drawedge(D,AD){$a$}
\drawedge(E,AD){$a$}
\drawedge(B,AB){$a$}
\drawedge(AC,BC){$a$}
\drawedge(BD,AC){$a$}
\drawedge(AE,AC){$a$}
\drawedge(BE,AE){$a$}
\drawedge(Y,BD){$a$}
\drawedge(X,AE){$a$}
\gasset{AHnb=0}
\drawedge[dash={1.05}0](A,C){}
\drawedge[dash={1.05}0](D,E){}
\drawedge[dash={1.05}0](E,CD){}
\drawedge[dash={1.05}0](X,BE){}
\drawedge[dash={1.05}0](D,CD){}
\drawedge[dash={1.05}0](B,AD){}
\drawedge[dash={1.05}0](BD,AE){}
\drawedge[dash={1.05}0,curvedepth=10](AB,AC){}
\end{picture}
\end{center}
\caption{A sample $a$-cluster. Marks are the $a$-heights of states; dashed lines show the equivalence $\ker a$.}\label{fig:cluster}
\end{figure}

We say that $a$ is a \emph{unimodal} letter if the defect of $a$ is positive and all states with maximal $a$-height have the same root. In the $a$-cluster shown in Fig.~\ref{fig:cluster}, the maximal $a$-height of a state is 3, and all states of this $a$-height  have the same root. If we assume that the $a$-heights of states in other $a$-clusters of $\Gamma_a$ are all less than 3, then the letter $a$ is unimodal.

Our first main result is the following.

\begin{thm}
\label{thm:unimodal}
Every primitive DFA possessing a unimodal letter is synchronizing.
\end{thm}

We prove Theorem~\ref{thm:unimodal} in Section~\ref{sec:proof}. Here we demonstrate some of its applications.

Given a DFA $\mA=\langle Q,\Sigma\rangle$, the \emph{kernel type} of a letter $a\in\Sigma$ is the non-increasing sequence of the class sizes of the equivalence $\ker a$. For an example, look again at the $a$-cluster shown in Fig.~\ref{fig:cluster}. If we assume that all other $a$-clusters of $\Gamma_a$ are cycles, then the letter $a$ has the kernel type $(3,2,2,2,2,2,1,1,\dots)$.

\begin{lem}\label{lem:kernel type}
In an arbitrary DFA, every letter of kernel type $(k,1,1,\dots)$, where $k>1$, is unimodal.
\end{lem}

\begin{proof}
Consider a DFA $\langle Q,\Sigma\rangle$ and fix a letter $a\in\Sigma$. For each state $p\in Q\dt a$, the set $p\dt a^{-1}:=\{q\in Q\mid q\dt a= p\}$ constitutes a class of the equivalence $\ker a$. Now suppose that the defect of $a$ is positive and a state $r\in Q$ is the root of some state $s\in Q$ in the graph $\Gamma_a$. Then $r=q\dt a=q'\dt a$ where $q$ is the predecessor of $r$ in the path that leads from $s$ to $r$ in $\Gamma_a$ and $q'$ is the predecessor of $r$ in the cycle of $\Gamma_a$ on which $r$ lies ($q'=r$ if the cycle degenerates into a loop). Since $q\ne q'$, the class $r\dt a^{-1}$ is non-singleton. We see that each root contributes an entry different from 1 to the kernel type of the letter $a$.

Thus, if the kernel type of $a$ has a unique entry different from 1, then the graph $\Gamma_a$ has only one root whence all states with positive $a$-heights have the same root. In particular, the unimodality condition holds. \end{proof}

Combining Lemma~\ref{lem:kernel type} and Theorem~\ref{thm:unimodal} immediately yields the following.

\begin{cor}\label{cor:kernel type}
Every primitive DFA possessing a letter of kernel type $(k,1,1,\dots)$, where $k>1$, is synchronizing.
\end{cor}

For permutation-primitive DFAs, synchronizability under the presence of a letter of kernel type $(k,1,1,\dots)$, $k>1$, was established by Jo\~ao Ara\'ujo and Peter Cameron~\cite{AC14}; see Theorem 2 therein. Corollary~\ref{cor:kernel type} generalizes this result.

In \cite{RS23}, Igor Rystsov and Marek Szyku\l{}a came up with two conjectures on the synchronizability of primitive DFAs. The first of them, restated in the terminology of the present note, is the following.
\begin{coni}
Every primitive DFA with all letters of defect at most $1$ is synchronizing unless all letters have defect $0$.
\end{coni}
In \cite[Section 3.3]{RS23}, Conjecture RS was supported by some experimental data; in particular, it was reported to hold for all DFAs with two letters and at most 11 states. Besides, it was shown in \cite[Section 3.1]{RS23} that several known results implied the validity of Conjecture RS in certain classes of DFAs.

Since the kernel type of a letter of defect 1 is $(2,1,1,\dots)$, the following result (validating Conjecture RS) is a special case of Corollary~\ref{cor:kernel type}.
\begin{cor}\label{cor:defect1}
Every primitive DFA possessing a letter of defect $1$ is synchronizing.
\end{cor}

One may ask whether primitivity implies synchronizability also in the presence of a letter of defect 2. The answer is negative as shown in \cite{RS23}: Example 7 there provides two non-synchronizing primitive DFAs each of which has five states and a letter of defect 2. The graphs of these DFAs are shown in Fig.~\ref{fig:PrimitiveNonSynchro} borrowed from~\cite{RS23} with the authors' permission. 
\begin{figure}[htb]
\begin{tikzpicture}[node distance=2cm,scale=1,every node/.style={transform shape},bend angle=30]
\node[state, inner sep=0pt] (q0) {$q_0$};
\node[state, inner sep=0pt] [right=of q0] (q1) {$q_1$};
\node[state, inner sep=0pt] [right=of q1] (q2) {$q_2$};
\node[state, inner sep=0pt] [below=of q2] (q3) {$q_3$};
\node[state, inner sep=0pt] [left=of q3] (q4) {$q_4$};

\draw (q0) edge[bend left=20] node[auto,midway]{$a$} (q1);
\draw (q1) edge[bend left=20] node[auto,midway]{$a$} (q0);
\draw (q4) edge node[auto,midway]{$a$} (q0);
\draw (q2) edge[bend left] node[auto,midway]{$a$} (q3);
\draw (q3) edge[loop, looseness=7,out=-20,in=20] node[swap,midway]{$a$} (q3);

\draw(q0) edge[loop, looseness=7,in=200,out=160] node[swap,midway]{$b$} (q0);
\draw(q1) edge node[auto,midway]{$b$} (q2);
\draw(q2) edge node[auto,midway,swap]{$b$} (q3);
\draw(q3) edge node[auto,midway]{$b$} (q4);
\draw(q4) edge node[auto,midway]{$b$} (q1);
\end{tikzpicture}
\begin{tikzpicture}[node distance=2cm,scale=1,every node/.style={transform shape},bend angle=10]
\node[state, inner sep=0pt] (q0) {$q_0$};
\node[state, inner sep=0pt] [above left=of q0] (q1) {$q_1$};
\node[state, inner sep=0pt] [above right=of q0] (q2) {$q_2$};
\node[state, inner sep=0pt] [below right=of q0] (q3) {$q_3$};
\node[state, inner sep=0pt] [below left=of q0] (q4) {$q_4$};
\draw(q0) edge node[auto,midway]{$a$} (q3);
\draw(q2) edge node[auto,midway]{$a$} (q3);
\draw(q3) edge node[auto,midway]{$a$} (q4);
\draw(q4) edge node[auto,midway]{$a$} (q0);
\draw(q1) edge node[auto,midway]{$a$} (q0);

\draw(q1) edge[bend left] node[auto,midway]{$b$} (q4);
\draw(q4) edge[bend left] node[auto,midway]{$b$} (q1);
\draw(q0) edge[bend left] node[auto,midway]{$b$} (q2);
\draw(q2) edge[bend left] node[auto,midway]{$b$} (q0);
\draw(q3) edge[loop,out=-20,in=20,looseness=6] node[midway,swap]{$b,c$} (q3);

\draw(q1) edge[bend left] node[auto,midway]{$c$} (q2);
\draw(q2) edge[bend left] node[auto,midway]{$c$} (q1);
\draw(q0) edge[loop,out=250,in=290,looseness=6] node[midway,swap]{$c$} (q0);
\draw(q4) edge[loop,out=160,in=200,looseness=6] node[midway,swap]{$c$} (q4);
\end{tikzpicture}
\caption{Two non-synchronizing primitive DFAs with a letter of defect 2 from \cite[Example 7]{RS23}.}\label{fig:PrimitiveNonSynchro}
\end{figure}

On the other hand, every permutation-primitive DFA with a letter of defect 2 is synchronizing~\cite[Theorem 3(a)]{AC14}.

Given a DFA $\langle Q,\Sigma\rangle$, a letter $a\in\Sigma$ of positive defect is called a \emph{semiconstant} if $q\dt a=q'\dt a$ for all $q,q'\notin Q\dt a$ and $p\dt a=p$ for all $p\in Q\dt a$. In other words, a semiconstant fixes every state in its image and sends all states outside the image to one particular state. By \cite[Theorem 25]{RS23}, every primitive DFA whose letters are either of defect 0 or semiconstants is synchronizing unless all letters have defect $0$. Obviously, the kernel type of a semiconstant of defect $d$ is $(d+1,1,1,\dots)$. Thus, Corollary~\ref{cor:kernel type} readily leads to a stronger fact.

\begin{cor}\label{cor:semiconstant}
Every primitive DFA possessing a semiconstant is synchronizing.
\end{cor}

\section{Proof of Theorem~\ref{thm:unimodal}}
\label{sec:proof}

We start with a reduction which was deduced in \cite[Section 3.1]{RS23} from a known property of primitive DFAs; see \cite[Proposition 8]{Th70} or \cite[Proposition 5.1]{St10}. Here, we provide a direct proof for the sake of self-containedness.

A DFA $\langle Q,\Sigma\rangle$ is called \emph{\scn} if for all $q,q'\in Q$ there exists a word $w\in\Sigma^*$ such that $q\dt w=q'$.

\begin{prop}
\label{prop:reduction}
Every primitive DFA with more than two states is either \scn\ or synchronizing.
\end{prop}

\begin{proof}
Let $\mA=\langle Q,\Sigma\rangle$ be a DFA. A non-empty subset $S\subseteq Q$ is said to be \emph{invariant} if $s\dt a\in S$ for all $s\in S$ and $a\in\Sigma$. Given an invariant subset $S$, consider the relation
\[
\rho_S:=\{(q,q')\in Q\times Q\mid q,q'\in S \text{ or } q=q'\}.
\]
It is known (and easy to verify) that $\rho_S$ is a congruence on $\mA$ for which $S$ is a class. Hence, in a primitive DFA, every invariant subset either is a singleton or coincides with the set of all states.

Assume that $\mA=\langle Q,\Sigma\rangle$ is a primitive DFA with more than two states. For each $q\in Q$, consider the set $\ov{q}:=\{q\dt w\mid w\in\Sigma^*\}$. Clearly, $\ov{q}$ is an invariant subset whence either $\ov{q}=\{q\}$ or $\ov{q}=Q$. If $\ov{q}=Q$ for all $q\in Q$, then $\mA$ is \scn. Suppose that there is a state $q_0$ such that $\ov{q_0}=\{q_0\}$. If $q_1$ is another state with $\ov{q_1}=\{q_1\}$, then $\{q_0,q_1\}$ is a 2-element invariant subset. We then must have $\{q_0,q_1\}=Q$, a contradiction. Hence, $\ov{q}=Q$ for all $q\in Q\setminus\{q_0\}$; in particular, for each $q\in Q\setminus\{q_0\}$, there exists a word $w_q\in\Sigma^*$ such that $q\dt w_q=q_0$.

We inductively construct a reset word for $\mA$, starting with $w_0:=\varepsilon$. If a word $w_i$ has already been constructed and $Q\dt w_i=\{q_0\}$, then $w_i$ is a reset word. Otherwise, take any $q\in Q\dt w_i\setminus\{q_0\}$ and let $w_{i+1}:=w_iw_q$. Since $q\dt w_q=q_0\dt w_q=q_0$, the cardinality of $Q\dt w_{i+1}$ is strictly less than that of $Q\dt w_i$. Hence, the described process eventually produces a reset word, and $\mA$ is synchronizing.
\end{proof}

We need a notion which is due to Karel Culik II, Juhani Karhum\"aki, and Jarkko Kari~\cite{CKK02}. They defined the \emph{stability relation} $\sigma$ on a DFA $\langle Q,\Sigma\rangle$ as follows:
\[
\sigma:=\{(q,q')\in Q\times Q\mid \forall v\in \Sigma^*\ \exists w\in \Sigma^*\ \ q\dt vw=q'\dt vw\}.
\]
The following properties were observed in~\cite{CKK02}.
\begin{lem}
\label{lem:stability}
\begin{enumerate}
  \item On each DFA, the stability relation is a congruence.
  \item A DFA is synchronizing if an only if its stability relation is universal.
\end{enumerate}
\end{lem}

The final ingredient we need stems from Avraham Trahtman's proof of the Road Coloring Conjecture in \cite{Tr09}. We use the presentation of Trahtman's argument given in \cite{KV}; the following is Lemma 4.4 from~\cite{KV} restated in the terminology adopted in the present note.

\begin{lem}
\label{lem:unimodal}
In a \scn\ DFA possessing a unimodal letter, the stability relation is not the equality.
\end{lem}

\begin{proof}[Proof of Theorem~\ref{thm:unimodal}]
Let $\mA$ be a primitive DFA possessing a unimodal letter; we aim to show that $\mA$ is synchronizing. If $\mA$ has at most two states, then the unimodal letter is easily seen to be a reset word for $\mA$. Otherwise, in view of Proposition~\ref{prop:reduction}, we may assume that $\mA$ is \scn. Then Lemma~\ref{lem:unimodal} ensures that the stability relation $\sigma$ on $\mA$ is not the equality. Since $\mA$ is primitive, Lemma~\ref{lem:stability}(1) implies that $\sigma$ is universal whence $\mA$ is synchronizing by Lemma~\ref{lem:stability}(2).
\end{proof}

\section{Primitive automata with a word of rank 2}
\label{sec:rank2}

Our second main result generalizes Peter Neumann's lemma on permutation-primitive DFAs  \cite[Lemma 2.4]{Ne09} that has already been mentioned in Section~\ref{sec:intro}.

\begin{thm}
\label{thm:rank2}
A primitive DFA $\mA=\langle Q,\Sigma\rangle$ with at least three states is synchronizing whenever some word $w\in\Sigma^*$ has rank $2$.
\end{thm}

\begin{proof}
Arguing by contradiction, assume that $\mA$ is not synchronizing. Then $\mA$ is \scn\ by Proposition~\ref{prop:reduction}.

The set $S:=Q\dt w$ consists of two states since $w$ has rank 2. Moreover, $|S\dt u|=2$ for every word $u\in\Sigma^*$ because if $|S\dt u|=1$ for some $u\in\Sigma^*$, then $wu$ would be a reset word for $\mA$ in a contradiction to the assumption that $\mA$ is not synchronizing. Let $S=\{s,s'\}$. Since $\mA$ is \scn, for every state $q\in Q$, there exists a word $u\in\Sigma^*$ such that $s\dt u=q$, that is, $q\in S\dt u$. Thus, $Q$ is a union of 2-element sets of the form $S\dt u$, $u\in\Sigma^*$. If different sets of this form are pairwise disjoint, then the relation
\[
\rho:=\{(q,q')\in Q\times Q\mid \exists u\in\Sigma^*\ q,q'\in S\dt u\}
\]
is easily seen to be a congruence on $\mA$. Since each $\rho$-class consists of two states and $\mA$ has at least three states, $\rho$ is not the universal relation, nor is $\rho$ the equality. This contradicts the primitivity of $\mA$.

Thus, there exist two words $u_1,u_2\in\Sigma^*$ such that the subsets $S_1:=S\dt u_1$ and $S_2:=S\dt u_2$ are different but have a common state, say, $p$. If $p_i$ is the other state in $S_i$, $i=1,2$, we have $p_1\ne p_2$. Take an arbitrary word $v\in\Sigma^*$ and consider the set $\{p,p_1,p_2\}\dt vw$. It is contained in the 2-element subset $S=Q\dt w$ whence amongst the three states $p\dt vw,p_1\dt vw,p_2\dt vw$, some two must be equal. We have $p\dt vw\ne p_i\dt vw$ since $\{p\dt vw,p_i\dt vw\}=S_i\dt vw=S\dt u_ivw$ for $i=1,2$, and each set of the form $S\dt u$, $u\in\Sigma^*$, consists of two states. The only remaining option is $p_1\dt vw=p_2\dt vw$.

We have thus proved that the pair $(p_1,p_2)$ with $p_1\ne p_2$ belongs to the stability relation $\sigma$ on $\mA$. Hence, $\sigma$ is not the equality, and since $\mA$ is primitive, Lemma~\ref{lem:stability}(1) implies that $\sigma$ is the universal relation. By Lemma~\ref{lem:stability}(2), $\mA$ is synchronizing, a contradiction.
\end{proof}

It was mentioned after Corollary~\ref{cor:defect1} that the presence of a letter of defect 2 forces synchronizability for permutation-primitive DFAs but fails to do so for general primitive DFAs. Continuing this line of discussion, one may ask whether primitivity implies synchronizability in the presence of a word or a letter of rank 3. The answer is negative even for the permutation-primitive case. An example of a non-synchronizing permutation-primitive DFA with nine states and a letter of rank 3 is described after Theorem 4 in~\cite{AC14}; the graph of this DFA is shown in Fig.~\ref{fig:grid}, where the loops have been omitted for better readability. 
\begin{figure}[htb]
\begin{tikzpicture}[node distance=2cm,scale=1,every node/.style={transform shape},bend angle=10]
\node[state, inner sep=0pt] (q00) {$q_{00}$};
\node[state, inner sep=0pt] [right=of q00] (q01) {$q_{01}$};
\node[state, inner sep=0pt] [right=of q01] (q02) {$q_{02}$};
\node[state, inner sep=0pt] [below=of q00] (q10) {$q_{10}$};
\node[state, inner sep=0pt] [right=of q10] (q11) {$q_{11}$};
\node[state, inner sep=0pt] [right=of q11] (q12) {$q_{12}$};
\node[state, inner sep=0pt] [below=of q10] (q20) {$q_{20}$};
\node[state, inner sep=0pt] [right=of q20] (q21) {$q_{21}$};
\node[state, inner sep=0pt] [right=of q21] (q22) {$q_{22}$};
\draw(q00) edge[in=80,out=280] node[auto,midway]{$b$} (q10);
\draw(q10) edge[in=260,out=100] node[auto,midway]{$b$} (q00);
\draw(q01) edge[in=80,out=280] node[auto,midway]{$b$} (q11);
\draw(q11) edge[in=260,out=100] node[auto,midway]{$b$} (q01);
\draw(q02) edge[in=80,out=280] node[auto,midway]{$b$} (q12);
\draw(q12) edge[in=260,out=100] node[auto,midway]{$b$} (q02);
\draw(q10) edge[in=80,out=280] node[auto,midway]{$c$} (q20);
\draw(q20) edge[in=260,out=100] node[auto,midway]{$c$} (q10);
\draw(q11) edge[in=80,out=280] node[auto,midway]{$c$} (q21);
\draw(q21) edge[in=260,out=100] node[auto,midway]{$c$} (q11);
\draw(q12) edge[in=80,out=280] node[auto,midway]{$c$} (q22);
\draw(q22) edge[in=260,out=100] node[auto,midway]{$c$} (q12);
\draw(q01) edge[bend left] node[auto,midway]{$e$} (q10);
\draw(q10) edge[bend left] node[auto,midway]{$e$} (q01);
\draw(q21) edge[bend left] node[auto,midway]{$e$} (q12);
\draw(q12) edge[bend left] node[auto,midway]{$e$} (q21);
\draw(q02) edge[bend angle=13,bend left] node[auto,midway]{$e$} (q20);
\draw(q20) edge[bend angle=13,bend left] node[auto,midway]{$e$} (q02);
\draw(q01) edge node[auto,midway]{$a$} (q00);
\draw(q02) edge[bend angle=18,bend right] node[swap,midway]{$a$} (q00);
\draw(q10) edge node[auto,midway]{$a$} (q11);
\draw(q12) edge node[auto,midway]{$a$} (q11);
\draw(q21) edge node[auto,midway]{$a$} (q22);
\draw(q20) edge[bend angle=18,bend right] node[swap,midway]{$a$} (q22);
\end{tikzpicture}
\caption{A non-synchronizing permutation-primitive DFA from \cite{AC14} with a letter of rank 3. If the action of a letter $x\in\{a,b,c,e\}$ at some state is not shown, then the loop labeled $x$ is assumed.}\label{fig:grid}
\end{figure}

Each of the two 5-state non-synchronizing primitive DFAs from Fig.~\ref{fig:PrimitiveNonSynchro} also has a letter of rank 3, but these DFAs are not permutation-primitive.

\begin{rmk}
\emph{We could have stated and proved Corollary~\ref{cor:defect1} in the form used in Theorem~\ref{thm:rank2}, that is, requesting the presence of a} word \emph{of defect 1 rather than a letter of defect 1. However, this would not be a generalization since a DFA having a word of defect 1 necessarily has a letter of defect 1. In contrast, a DFA can have a word of rank 2 without having any letter of rank 2.}
\end{rmk}

\subsection*{Acknowledgments} This note was inspired by the conjectures and findings presented in the preprint \cite{RS23} by Igor Rystsov and Marek Szyku\l{}a. I would like to express my gratitude to them for sharing their unpublished manuscript with me before it was posted on arXiv.

My research was supported by the Ministry of Science and Higher Education of the Russian Federation, project FEUZ-2023-0022.

\end{document}